\documentclass[12pt,twoside]{article} 
\usepackage[T1]{fontenc} 

\usepackage{amsthm, amsmath, amssymb, mathrsfs, upgreek, enumerate, mathtools,
  comment, natbib, subfigure, graphicx, mathabx, booktabs, threeparttable,
tikz, relsize, exscale, epstopdf, scrextend, pdflscape, bm, multirow}
\usepackage[headheight=110pt,margin=1.25in]{geometry}
\usepackage[algoruled]{algorithm2e}
\usepackage[section]{placeins} 

\usepackage{fancyhdr, setspace} 
\usepackage[pagebackref=false]{hyperref} 
\hypersetup{
    colorlinks=true,
    citecolor=blue,
    filecolor=black,
    linkcolor=black,
    urlcolor=blue,
    bookmarksopen=true,
    pdfstartview=FitH
}

\newtheorem{theorem}{Theorem}
\newtheorem{theoremSA}{Theorem}[section]

\newtheorem{assump*}{Assumption}[section]
\newtheorem{lemma}{Lemma}

\theoremstyle{definition}

\newtheorem{remark}{Remark}

\def\equationautorefname~#1\null{(#1)\null} 

\def\Snospace~{\S{}}

\makeatletter
\def\thm@space@setup{
  \thm@preskip=15pt \thm@postskip=15pt 
}
\makeatother

\newcommand{\var}{\text{Var}}
\newcommand{\E}{{\bf E}}
\newcommand{\R}{\mathbb{R}}
\newcommand{\Sv}{\bm{S}}
\newcommand{\prob}{\mathbf{P}}
\newcommand{\plimarrow}{\stackrel{p}\longrightarrow}
\newcommand{\dlimarrow}{\stackrel{d}\longrightarrow}
\newcommand{\ind}{\mathbf{1}}

\newcommand*{\medcap}{\mathbin{\scalebox{1.5}{\ensuremath{\cap}}}}

\providecommand{\abs}[1]{\lvert#1\rvert} 
\providecommand{\norm}[1]{\lVert#1\rVert}

\renewcommand{\qed}{\hfill \mbox{\raggedright \rule{0.08in}{0.08in}}} 
\renewenvironment{proof}[1][\proofname]{{\noindent\sc#1. }}{\qed\vspace{15pt}} 

\title{\bf\sc Dependence-Robust Inference Using Resampled Statistics\thanks{I thank the editor and referees for comments that improved the quality of the paper. I also thank Eric Auerbach, Vittorio Bassi, Jinyong Hahn, Roger Moon, Hashem Pesaran, Kevin Song, and seminar audiences at UC Davis, USC INET, the 2018 Econometric Society Winter Meetings, the 2018 Econometrics Summer Masterclass and Workshop at Warwick, and the NetSci2018 Satellite on Causal Inference and Design of Experiments. This research is supported by NSF Grant SES-1755100.}}

\author{Michael P.\ Leung\thanks{Department of Economics, University of Southern California. E-mail: leungm@usc.edu.}} 

\begin{document}
\maketitle

\begin{abstract}

  {\sc Abstract.} We develop inference procedures robust to general forms of weak dependence. The procedures utilize test statistics constructed by resampling in a manner that does not depend on the unknown correlation structure of the data. We prove that the statistics are asymptotically normal under the weak requirement that the target parameter can be consistently estimated at the parametric rate. This holds for regular estimators under many well-known forms of weak dependence and justifies the claim of dependence-robustness. We consider applications to settings with unknown or complicated forms of dependence, with various forms of network dependence as leading examples. We develop tests for both moment equalities and inequalities.
  
  \vspace{15pt}


 
\end{abstract}

\addcontentsline{toc}{part}{Main Paper}
\newpage

\section{Introduction}
\onehalfspacing

This paper builds on randomized subsampling tests due to \cite{song_ordering-free_2016} and proposes inference procedures for settings in which the dependence structure of the data is complex or unknown. This is useful for a variety of applications using, for example, network data, clustered data when cluster memberships are imperfectly observed, or spatial data with unknown locations. The proposed procedures compare a test statistic, constructed using a set of resampled observations, to a critical value constructed either using a normal approximation or by resampling. Computation is the same regardless of the dependence structure. We prove that our procedures are asymptotically valid under the weak requirement that the target parameter can be consistently estimated at the $\sqrt{n}$ rate, a condition satisfied by most forms of weakly dependent data for regular estimators. In this sense, inference using our resampled statistics is robust to quite general forms of weak dependence. 

Consider the simple problem of inference on a scalar population mean $\mu_0$. Typically we assume that the sample mean $\bar{X}$ is asymptotically normal in the sense that
\begin{equation*}
  \sqrt{n}(\bar{X} - \mu_0) \dlimarrow \mathcal{N}(0, \bm{\Sigma}).
\end{equation*}

\noindent However, in a setting with complex or unknown forms of dependence, it may be difficult or unclear how to estimate $\bm{\Sigma}$ due to the covariance terms. A simple statistic we utilize for inference in this setting is
\begin{equation*}
  \tilde T_M = \sqrt{R_n} \hat{\bm{\Sigma}}^{-1/2} (\bar{X}^* - \mu_0),
\end{equation*}

\noindent where $\bar{X}^*$ is the mean of $R_n$ draws with replacement from the data and $\hat{\bm{\Sigma}} = n^{-1} \sum_{i=1}^n (X_i-\bar{X})^2$ is the (naive) sample variance. Using the identity
\begin{equation*}
  \tilde T_M = \underbrace{\tilde T_M - \E[\tilde T_M \mid \bm{X}]}_{[I]} + \underbrace{\E[\tilde T_M \mid \bm{X}]}_{[II]},
\end{equation*}

\noindent where $\bm{X}$ is the data, we show that $[I] \dlimarrow \mathcal{N}(0,1)$ and the bias term $[II]$ is asymptotically negligible if $R_n$ is chosen to diverge at a sufficiently slow rate. We can therefore compare $\tilde T_M$ to a normal critical value to conduct inference on $\mu_0$ regardless of the underlying dependence structure. As we later discuss, larger values of $R_n$ generate higher power through a faster rate of convergence for $[I]$ but also a larger bias term $[II]$, which generates size distortion. Thus for practical implementation, we suggest a rule of thumb for $R_n$ that accounts for this trade-off.

\cite{song_ordering-free_2016} proposes tests for moment equalities when the data satisfies a particular form of weak dependence known as ``local dependence'' in which the dependence structure is characterized by a graph. The novelty of his procedure is that its implementation and asymptotic validity does not require knowledge of the graph. We study essentially the same test statistic but generalize his theoretical results, showing asymptotic validity under the substantially weaker requirement of estimability at a $\sqrt{n}$ rate, which significantly broadens the applicability of the method. Additionally, we propose new methods for testing moment inequalities.

Many resampling methods are available for spatial, temporal, and clustered data when the dependence structure is known \citep[e.g.][]{cameron2008bootstrap,lahiri2013resampling,politis_subsampling_1999}. These methods are used to construct critical values for a test statistic computed with the original dataset $\bm{X}$. For the critical values to be asymptotically valid, resampling is implemented in a way that mimics the dependence structure of the data, commonly by redrawing blocks of neighboring observations, but this requires knowledge of the dependence structure. In contrast, our procedures involve computing a resampled test statistic and critical values based on its limiting distribution conditional on the data. Hence, there is no need to mimic the actual dependence structure, which is why our procedures can be dependence-robust. 

Of course, the broad applicability of our procedures comes at a cost. The main drawback is inefficiency due to the fact that the test statistics are essentially computed from a subsample of observations. In contrast, the test statistic under conventional subsampling, for example, utilizes the full sample. Thus, in settings where inference procedures exist, our resampled statistics suffer from slower rates of convergence, which yield tests with lower power and can exacerbate finite-sample concerns such as weak instruments. We interpret this as the cost of dependence robustness. Our objective is not to propose a procedure that is competitive with existing procedures but rather to provide a broadly applicable and robust inference procedure that can be useful when little is known about the dependence structure or when this structure is complex and no inference procedure is presently available.

We consider four applications. The first is regression with unknown forms of weak dependence. This setting is relevant when the dependent or independent variables are functions of a social network, as in network regressions \citep{chandrasekhar2016econometrics}. Another special case is cluster dependence when the level of clustering is unknown and the number of clusters is small, settings in which conventional clustered standard errors can perform poorly \citep{cameron2015practitioner}. The second application is estimating treatment spillovers on a partially observed network. The third is inference on network statistics, a challenging setting because different network formation models induce different dependence structures. The fourth application is testing for a power law distribution, a problem that has received a great deal of attention in economics, network science, biology, and physics \citep{barabasi1999emergence,gabaix2009power,newman2005power}. Widely used methods in practice assume that the underlying data is i.i.d.\ \citep{clauset2009power,klaus2011statistical}, which is often implausible in applications involving spatial, financial, or network data.

The outline of the paper is as follows. The next section introduces our inference procedures. We then discuss four applications in \autoref{stapp}, followed by an empirical illustration on testing for power law degree distributions in \autoref{sapp}.  Next, \autoref{smain} states formal results on the asymptotic validity of our procedures. In \autoref{ssims}, we present simulation results from four different data-generating processes. Lastly, \autoref{sconclude} concludes. 

\section{Overview of Methods}\label{smodel}

We begin with a description of our proposed inference procedures. Throughout, let $\bm{X} = \{X_i\}_{i=1}^n \subseteq \R^m$ be a set of $n$ identically distributed random vectors with possibly dependent row elements. Denote the sample mean of $\bm{X}$ by $\bar{X}$. The goal is to conduct inference on some parameter $\mu_0 \in \R^m$. A simple example is the population mean $\mu_0 = \E[X_1]$, but we will also consider other parameters when discussing asymptotically linear estimators. Our main assumption will require $\bm{X}$ to be weakly dependent in the sense that $\bar{X}$ is $\sqrt{n}$-consistent for $\mu_0$.

\subsection{Moment Equalities}\label{smomeq}

We first consider testing the null hypothesis that $\mu_0 = \mu$ for some $\mu \in \R^m$ and constructing confidence regions for $\mu_0$. Let $R_n \geq 2$ be an integer and $\Pi$ the set of all bijections (permutation functions) on $\{1, \dots, n\}$. Let $\{\pi_r\}_{r=1}^{R_n}$ be a set of $R_n$ i.i.d.\ uniform draws from $\Pi$ and $\pi = (\pi_r)_{r=1}^{R_n}$. Define the sample variance matrix $\hat{\bm{\Sigma}} = n^{-1}\sum_{i=1}^n (X_i - \bar{X})(X_i - \bar{X})'$. 

We focus on two test statistics. The first is the {\em mean-type statistic}, given by
\begin{equation*}
  T_M(\mu; \pi) = \tilde T_M(\mu;\pi)' \tilde T_M(\mu;\pi), \quad\text{where}\quad \tilde T_M(\mu;\pi) = \frac{1}{\sqrt{R_n}} \sum_{r=1}^{R_n} \hat{\bm{\Sigma}}^{-1/2}\left(X_{\pi_r(1)} - \mu\right).
\end{equation*}

\noindent That is, $\tilde T_M(\mu;\pi)$ is computed by drawing $R_n$ observations with replacement from $\{\hat{\bm{\Sigma}}^{-1/2}(X_i - \mu)\}_{i=1}^n$ and then taking the average and scaling up by $\sqrt{R_n}$. Note that we compute $\hat{\bm{\Sigma}}$ using the full sample.
 
The second test statistic is the {\em U-type statistic}, which is given by
\begin{equation*}
  T_U(\mu; \pi) = \frac{1}{\sqrt{m R_n}} \sum_{r=1}^{R_n} (X_{\pi_r(1)}-\mu)' \hat{\bm{\Sigma}}^{-1} (X_{\pi_r(2)}-\mu) 
\end{equation*}

\noindent and essentially follows \cite{song_ordering-free_2016}. Unlike the mean-type statistic, here we resample pairs of observations with replacement and compute a quadratic form. 

\bigskip

\noindent {\bf Inference Procedures.} We prove that if $\bar{X}$ is $\sqrt{n}$-consistent for $\mu_0$, then
\begin{align}
  \begin{split}\label{goal}
  &\tilde T_M(\mu_0; \pi) \dlimarrow \mathcal{N}(\bm{0},\bm{I}_m) \quad\text{if}\quad R_n/n \rightarrow 0 \quad\text{and} \\
  &T_U(\mu_0; \pi) \dlimarrow \mathcal{N}(0,1) \quad\text{if}\quad \sqrt{R_n}/n \rightarrow 0
  \end{split} 
\end{align} 

\noindent (under regularity conditions), where $\bm{I}_m$ is the $m\times m$ identity matrix. CLTs for a wide range of notions of weak dependence, including mixing, near-epoch dependence, various forms of network dependence, etc., can be used to verify $\sqrt{n}$-consistency. There are some examples of dependent data that violate $\sqrt{n}$-consistency. One is cluster dependence with many clusters, large cluster sizes, and strongly dependent observations within clusters \citep{hansen2019asymptotic}. We discuss in \autoref{slower} below how other rates of convergence can be accommodated by adjusting $R_n$.

Result \eqref{goal} enables us to construct critical values for testing. For example, to test the null that $\mu_0 = \mu$ against a two-sided alternative, we can use
\begin{equation}
   \ind\{T_U(\mu; \pi) > z_{1-\alpha} \} \quad\text{or}\quad \ind\{T_M(\mu; \pi) > q_{1-\alpha} \}, \label{tests}
\end{equation}

\noindent where $z_{1-\alpha}$ and $q_{1-\alpha}$ are respectively the $(1-\alpha)$-quantiles of the standard normal distribution and chi-square distribution with $m$ degrees of freedom. Note that the U-type statistic is two-sided in nature, which is why in \eqref{tests} we simply compare $T_U(\mu; \pi)$, rather than its absolute value, to a normal quantile. To test one-sided alternatives with the U-type statistic, we can additionally exploit the sign of $\bar{X}$, as done in the application in \autoref{stpl} and the moment inequality test below. For instance, we can choose to reject only if the test statistic exceeds its critical value and the sign of $\bar{X}$ is positive.

In the case of scalar data ($m=1$), we obtain the following simple confidence interval (CI) for $\mu_0$ using the mean-type statistic:
\begin{equation}
  \frac{1}{R_n} \sum_{r=1}^{R_n} X_{\pi_r(1)} \pm z_{1-\alpha/2} \frac{\hat{\bm{\Sigma}}^{1/2}}{\sqrt{R_n}}. \label{MtypeCI}
\end{equation}

\noindent Alternatively, we can use the U-type statistic to obtain a CI by test inversion. 

\bigskip

\noindent {\bf Why This Works.} To see the intuition behind \eqref{goal}, consider the mean-type statistic. Define $W_{M,r} = \hat{\bm{\Sigma}}^{-1/2}(X_{\pi_r(1)} - \mu_0)$. Then
\begin{equation}
  \tilde T_M(\mu_0;\pi) = \underbrace{\frac{1}{\sqrt{R_n}} \sum_{r=1}^{R_n} \left( W_{M,r} - \E[W_{M,r} \mid \bm{X}] \right)}_{[I]} + \underbrace{\frac{1}{\sqrt{R_n}} \sum_{r=1}^{R_n} \E[W_{M,r} \mid \bm{X}]}_{[II]}. \label{decomp}
\end{equation}

\noindent Some algebra shows that $[II] = (R_n/n)^{1/2} \hat{\bm{\Sigma}}^{-1/2} n^{-1/2} \sum_{i=1}^n (X_i-\mu_0)$. Since $R_n/n = o(1)$ and $n^{-1/2} \sum_{i=1}^n (X_i-\mu_0) = O_p(1)$ under the assumption of $\sqrt{n}$-consistency, we have $[II] = o_p(1)$, provided the sample variance converges to a positive-definite matrix. Since the random permutations are i.i.d.\ conditional on $\bm{X}$, $[I] \dlimarrow \mathcal{N}(\bm{0},\bm{I}_m)$. The proof for $T_U(\mu_0; \pi)$ follows a similar logic.

\bigskip

\noindent {\bf Choice of $R_n$ and Statistic.} When choosing $R_n$, we face the following trade-off. A larger value corresponds to using a larger number of observations to construct the test statistic, which translates to higher power through a faster rate of convergence for part $[I]$ of decomposition \eqref{decomp}. On the other hand, a smaller value ensures that the bias term $[II]$ in \eqref{decomp} is negligible, which is important for size control. As is clear from the previous proof sketch, for the mean-type statistic, the bias and rate of convergence are respectively of order $(R_n/n)^{1/2}$ and $R_n^{-1/2}$. For the U-type statistic, they are instead $\sqrt{R_n}/n$ and $R_n^{-1/4}$, as discussed in \autoref{power}. For both statistics, we propose choosing $R_n$ to minimize the sum of these terms, yielding
\begin{equation}
  R_n^M = \sqrt{n} \quad\text{and}\quad R_n^U = (n/2)^{4/3}
  \label{R}
\end{equation}

\noindent for the mean- and U-type statistics, respectively. The choice of minimizing the sum of the two is only a heuristic but at least reflects an asymptotic trade-off between control of type I and type II errors. We only seek to provide a practical recommendation that accounts in some way for the trade-off and leave more sophisticated and data-dependent choices of $R_n$ to future work. We also note that simulation results in \autoref{ssims} show that the tests perform similarly for a variety of choices of $R_n$ around \eqref{R}.

The rates of convergence of the mean- and U-type statistics when choosing $R_n$ according to \eqref{R} are respectively $n^{-1/4}$ and $n^{-1/3}$. In general, the U-type statistic has better power properties, as shown theoretically in \cite{song_ordering-free_2016} in the context of locally dependent data. Our simulations confirm this for \eqref{R} across a wider range of dependence structures, which leads us to recommend use of the U-type over the mean-type statistic for smaller samples. The main appeal of the mean-type statistic is the ease of constructing CIs, as in \eqref{MtypeCI}.

\bigskip

\noindent {\bf Asymptotically Linear Estimators.} Suppose we observe identically distributed data $\bm{Z} = \{Z_i\}_{i=1}^n$ and are interested in a parameter $\beta_0 \in \R^d$. Suppose there exists a parameter $\theta_0$ and a function $\psi$ satisfying $\E[\psi(Z_1; \beta_0, \theta_0)] = 0$, and let $\hat\theta$ be an estimate of $\theta_0$.  Consider an estimator $\hat\beta$ that is asymptotically linear in the sense that
\begin{equation*}
  \sqrt{n}(\hat\beta - \beta_0) = \frac{1}{\sqrt{n}} \sum_{i=1}^n \psi(Z_i; \beta_0, \hat\theta) + o_p(1) 
\end{equation*}

\noindent For example, in the case of maximum likelihood, $\hat\theta$ is the sample Hessian, and $\psi$ is the score function times the Hessian. We can then apply our procedures to conduct inference on $\beta_0$ by defining
\begin{equation}
  X_i = \psi(Z_i; \beta_0, \hat\theta) \label{ale2}
\end{equation}

\noindent and $\mu_0 = \E[\psi(Z_1; \beta_0, \theta_0)] = 0$. Note that in this example, $\mu_0$ is not the population mean of the ``data'' $\bm{X}$. As discussed in \autoref{firststage} below, under regularity conditions, our procedures are asymptotically valid if $(\hat\beta,\hat\theta)$ is $\sqrt{n}$-consistent for $(\beta_0,\theta_0)$.

\begin{remark}
  As pointed out to us by Eric Auerbach, an alternative dependence-robust test (here for the case $m=1$) is to reject when $\abs{\mathcal{Z} + \sqrt{n}(\bar{X}-\mu_0)/h_n} > z_{1-\alpha/2}$, where $\mathcal{Z} \sim \mathcal{N}(0,1)$ is independent of $\bm{X}$ and $h_n$ is a diverging sequence. Since $h_n$ is eventually larger than the asymptotic variance, the second term vanishes, and the test has asymptotic size $\alpha$ under the null. Under the alternative, the power tends to one at rate $\sqrt{n}/h_n$, which is always slower than $\sqrt{n}$. Thus, this test has similar power properties to our test, and the first-order asymptotics does not distinguish between them. Nonetheless, we do not view this test as a serious practical alternative. Choosing the tuning parameter $h_n$ literally corresponds to choosing the size of the standard error, which is clearly problematic in practice. Indeed, for any fixed choice of $h_n$, this test is almost always either conservative or anti-conservative. In contrast, our conditions suggest that for the U-type statistic (for example), $R_n$ should not be chosen larger than $n^2$ for the claim of size control to be considered credible in finite samples, given that $\sqrt{R_n}/n\rightarrow 0$ is required for asymptotic validity. We also provide guidance for choosing $R_n$ in practice \eqref{R} and validate this choice across a wide range of dependence structures in extensive simulation experiments.
\end{remark}

\begin{remark}\label{nonrandom}
  Since the test statistics are random conditional on the data due to the permutation draws $\pi$, different researchers can reach different conclusions with the same dataset. This occurs with small probability for $n$ large, but for smaller samples, it is useful to have a procedure less sensitive to $\pi$. In his \S 3.5, \cite{song_ordering-free_2016}, proposes a procedure that allows the researcher to make the influence of $\pi$ as small as desired, which we reproduce here. Let $\{\tilde\pi_{r\ell}\colon \ell \,{=}\, 1,\dots,L;\, r \,{=}\, 1,\dots,R_n\}$ be i.i.d.\ uniform draws from $\Pi$ and $\tilde\pi_\ell = (\tilde\pi_{r\ell})_{r=1}^{R_n}$. Define the ``randomized confidence function''
  \begin{equation*}
    f_L(\mu; \alpha) = \frac{1}{L} \sum_{\ell=1}^L \ind\left\{ T_U(\mu; \tilde\pi_\ell) \leq z_{1-\alpha} \right\}.
  \end{equation*}

  \noindent (For the mean-type statistic, we instead use $T_M(\mu; \tilde\pi_\ell)$ and $q_{1-\alpha}$ in place of $T_U(\mu; \tilde\pi_\ell)$ and $z_{1-\alpha}$, respectively.) By taking $L$ as large as desired, we can make $f_L(\mu; \alpha)$ arbitrarily close to a nonrandom function of the data by the law of large numbers, which solves the randomness problem. To see how this function can be used for inference, for any small $\beta \in (0,\alpha)$ chosen by the econometrician, define the confidence region
  \begin{equation*}
    C_L(\alpha; \beta) = \left\{ \mu \in \R^m\colon f_L(\mu; \alpha-\beta) \geq 1-\alpha \right\}.
  \end{equation*}

  \noindent Using \eqref{goal}, it is straightforward to show that $\lim_{n\rightarrow\infty} \lim_{L\rightarrow\infty} \prob(\mu_0 \in C_L(\alpha; \beta)) \geq 1-\alpha$, so the confidence region has the desired asymptotic coverage. For the case of locally dependent data, this follows from Corollary 3.1 of \cite{song_ordering-free_2016}. It immediately generalizes to other forms of weak dependence by applying our \autoref{mq}.
\end{remark}

\subsection{Moment Inequalities}

We next consider testing the null $\mu_0 \leq 0$ for $\mu_0 = \E[X_1]$, where ``$\leq$'' denotes component-wise inequality. This is relevant, for example, for inference in strategic models of network formation \citep{sheng2014} and models of social interactions \citep{li2016partial}. Let $T_{U,k}(\mu_k; \pi)$ be the U-type statistic applied to scalar data $\{X_{i,k}\}_{i=1}^n$, where $X_{i,k}$ is the $k$th component of $X_i$ and $\mu_k \in \R$. Also let $\bar{X}_k$ be the $k$th component of $\bar{X}$ and $\hat{\bm{\Sigma}}_{kk}$ the $k$th diagonal of $\hat{\bm{\Sigma}}$. We propose the test statistic
\begin{align*}
  &Q_n(\pi) = \max_{1\leq k\leq m} \big\{ T_{U,k}(0; \pi) - \hat\lambda_k\ind\{\bar{X}_k < 0\} \big\}, \quad\text{where} \\
  &\hat\lambda_k = \bar{X}_k \hat{\bm{\Sigma}}_{kk}^{-1} \frac{1}{\sqrt{mR_n}} \sum_{r=1}^{R_n} \left( X_{\pi_r(1),k} + X_{\pi_r(2),k} \right) - \sqrt{\frac{R_n}{m}} \hat{\bm{\Sigma}}_{kk}^{-1} \bar{X}_k^2.
\end{align*}

\noindent While complicated in appearance, this is easy to compute in practice. Define
\begin{equation}
  \tilde Q_n(\pi) = \max_{1\leq k\leq m} T_{U,k}(\bar{X}_k;\pi), \label{tildeQn}
\end{equation} 

\noindent and let $c_{1-\alpha}$ be the $(1-\alpha)$-quantile of the conditional-on-$\bm{X}$ distribution of $\tilde Q_n(\pi)$. Our proposed test is to reject if and only if $\phi_n = 1$ for
\begin{equation}
  \phi_n = \ind\{ Q_n(\pi) > c_{1-\alpha} \}. \label{mitest}
\end{equation}

\noindent In practice, we can approximate $c_{1-\alpha}$ arbitrarily well by resampling $\pi$ $L$ times, computing $\tilde Q_n(\pi)$ for each draw, and then taking the appropriate sample quantile of this set of statistics. Formally, let $\tilde\pi_\ell = (\tilde\pi_{r\ell})_{r=1}^{R_n}$, where $\{\tilde\pi_{r\ell}\colon \ell \,{=}\, 1,\dots,L;\, r \,{=}\, 1,\dots,R_n\}$ are i.i.d.\ uniform draws from $\Pi$. The feasible critical value is
\begin{equation}
  c_{L,1-\alpha} = \inf\left\{ c'>0\colon \frac{1}{L} \sum_{\ell=1}^L \ind\left\{ \tilde Q_n(\tilde\pi_\ell) > c'\right\} \leq \alpha \right\}. \label{thepermCV}
\end{equation} 

In \autoref{mils}, we show that the proposed test uniformly exactly controls size. The intuition behind the test is as follows. Some algebra shows that
\begin{equation}
  Q_n(\pi) = \max_{1\leq k\leq m} \big\{T_{U,k}(\bar{X}_k; \pi) + \hat\lambda_k\ind\{\bar{X}_k \geq 0\}\big\}, \label{algebra}
\end{equation}

\noindent which is similar to $\tilde Q_n(\pi)$, except for the presence of the $\hat\lambda_k\ind\{\bar{X}_k \geq 0\}$ term. The indicator serves to detect the sign of $\mu_{0k}$, the $k$th component of $\mu_0$, and thus give the test power. To see this, first note that, in the appendix, we show in \eqref{inputsplit} and \eqref{RXfact} in the proof of \autoref{RSMI} that, for any $k$, $\hat\lambda_k \approx (R_n^{1/4}\mu_{0k})^2m^{-1/2}\hat{\bm{\Sigma}}^{-1}_{kk} + \mu_{0k}\cdot O_p(1)$, and $R_n^{1/4}\bar{X}_k \approx R_n^{1/4}\mu_{0k}$. We apply these two ``facts'' to the case $m=1$ for illustration.

First, under a fixed null, $\hat\lambda_k\ind\{\bar{X}_k \geq 0\} = \hat\lambda_k\ind\{R_n^{1/4}\bar{X}_k \geq 0\} \approx 0$ since either $R_n^{1/4}\mu_{0k} = 0$, in which case $\hat\lambda_k \approx 0$ by the first fact above, or $R_n^{1/4}\mu_{0k} < 0$, in which case the indicator is eventually zero by the second fact above. Consequently, $Q_n(\pi)$ and $\tilde Q_n(\pi)$ have the same asymptotic distributions, and the test controls size. Second, under the alternative, $R_n^{1/4}\bar{X}_k$ is instead eventually positive, so $\hat\lambda_k\ind\{R_n^{1/4}\bar{X}_k \geq 0\} \approx \hat\lambda_k$, which is positive with high probability. Indeed, for a fixed alternative, $\hat\lambda_k$ diverges by the first fact above and therefore so does $Q_n(\pi)$. On the other hand, $\tilde Q_n(\pi)$ has a non-degenerate limit distribution, so the test is consistent.

\begin{remark}\label{MIasymp}
  For the case $m=1$, dropping the subscript $k$, we have directly by \autoref{quadstat} in the appendix that $\tilde Q_n(\pi) \dlimarrow \mathcal{N}(0,1)$. Then we can use the test 
  \begin{equation*}
    \tilde\phi_n = \ind\{Q_n(\pi) - \hat\lambda\ind\{\bar{X} < 0\} > z_{1-\alpha}\}
  \end{equation*}

  \noindent which is similar to \eqref{mitest}, except it uses the computationally simpler asymptotic critical value $z_{1-\alpha}$, rather than the permutation critical value $c_{1-\alpha}$. 
\end{remark}

\section{Applications}\label{stapp}

\subsection{Regression with Unknown Weak Dependence}\label{unknown}

Let $\bm{D}$ be an $n\times k$ matrix of covariates with $i$th row $D_i$, and
\begin{equation*}
  Y_i = D_i'\beta_0 + \varepsilon_i, 
\end{equation*}

\noindent where $\{(D_i, \varepsilon_i)\}_{i=1}^n$ is identically distributed but possibly dependent. Our goal is inference on the $j$th component of $\beta_0$, denoted by $\beta_{0j}$. Common sources of dependence are clustering \citep{bertrand2004much}, spatial autocorrelation \citep{barrios2012clustering,bester2009inference}, and network autocorrelation \citep{acemoglu2015state}. Often the precise form of dependence may be unknown, or there may be insufficient data to compute conventional standard errors, for example if we do not fully observe the clusters, the spatial locations of the observations, or the network. If the OLS estimator is $\sqrt{n}$-consistent, however, inference on $\beta_{0j}$ is possible using resampled statistics.

In the context of network applications, a common exercise is to regress an outcome on some measure of the network centrality of a node \citep{chandrasekhar2016econometrics}. Such measures are inherently correlated across nodes even when links are i.i.d. Moreover, different models of network formation can lead to different expressions for the asymptotic variance. Because this model is an unknown nuisance parameter, it is useful to have an inference procedure that is valid for a large class of models.

To apply our procedure, we write the estimator in the form of a sample mean. Let $W_{ji}$ be the $ji$th component of the matrix $n (\bm{D}'\bm{D})^{-1} \bm{D}'$. Then 
\begin{equation*}
  \hat\beta_j = \frac{1}{n} \sum_{i=1}^n W_{ji}Y_i. 
\end{equation*}

\noindent is the OLS estimator for $\beta_{0j}$. This fits into the setup of \eqref{ale2} if we define $Z_i = (Y_i, D_i)$, $\hat\theta = n^{-1}\bm{D}'\bm{D}$, and $\psi(Z_i; \beta_0, \hat\theta) = W_{ji} Y_i - \beta_{0j}$. Thus, when computing our test statistics, we resample elements of $\bm{X}$ for $X_i = W_{ji}Y_i - \beta_{0j}$.

The main assumption required for the asymptotic validity of our procedure is $\sqrt{n}$-consistency of the estimators. For cluster dependence, this holds under conventional many-cluster asymptotics, where the number of observations in each cluster is small but the number of clusters is large. However, unlike with clustered standard errors, we need not know the right level of clustering or even observe cluster memberships. We can also allow the number of clusters to be small, possibly equal to one, so long as the data is weakly dependent within clusters in the sense of yielding $\sqrt{n}$-consistency. Within-cluster weak dependence is also required by \cite{canay2017randomization}, \cite{ibragimov2010t}, and \cite{ibragimov2016inference}, who propose novel inference procedures for cluster dependence with a small number of clusters. Resampled statistics are advantageous because we can allow for only a single cluster and do not require knowledge of cluster memberships.

For spatial dependence, $\sqrt{n}$-consistency can be established using CLTs for mixing or near-epoch dependent data \citep{jenish2012spatial}. CLTs for network statistics are mentioned in \autoref{sns}.

\subsection{Treatment Effects with Spillovers}\label{stes}

Suppose we observe data from a randomized experiment on a single network, where for each node $i$, we observe an outcome $Y_i$, a binary treatment assignment $D_i$, the number of nodes connected to $i$ (``network neighbors'') $\gamma_i$, and the number of treated network neighbors $T_i$. Consider the following outcome model studied in \cite{leung2017treatment}:
\begin{equation*}
  Y_i = r\left(D_i, T_i, \gamma_i, \varepsilon_i\right).
\end{equation*}

\noindent This departs from the conventional potential outcomes model by allowing $r(\cdot)$ to depend on $T_i$, which violates the stable unit treatment value assumption. The object of interest is the following measure of treatment/spillover effects:
\begin{equation}
  \E\left[ r(d,t,\gamma, \varepsilon_1) \mid \gamma_1=\gamma \right] - \E\left[ r(d',t',\gamma, \varepsilon_1) \mid \gamma_1=\gamma \right], \label{tse}
\end{equation}

\noindent where $t,t'\leq \gamma \in \mathbb{N}$. Variation across $d,d'$ identifies the direct causal effect of the treatment, while variation across $t,t'$ identifies a spillover effect, conditional on the number of neighbors. \cite{leung2017treatment} provides conditions on the network and dependence structure of $\{\varepsilon_i\}_{i=1}^n$ under which the sample analog of \eqref{tse} is $\sqrt{n}$-consistent. For example, we can allow $\varepsilon_i$ and $\varepsilon_j$ to be correlated if $i$ and $j$ are connected. Therefore, the setup falls within the scope of our assumptions.

Suppose the econometrician obtains data $\{W_i\}_{i=1}^n$ for $W_i = (Y_i,D_i,T_i,\gamma_i)$ by snowball-sampling 1-neighborhoods. That is, she first obtains a random sample of units, from which she gathers $(Y_i,D_i)$, and then she obtains the network neighbors of those units and their treatment assignment, from which she gathers $(T_i,\gamma_i)$. This is a very common method of network sampling. However, standard error formulas provided by \cite{leung2017treatment} may require knowledge of the path distances between observed units, which are usually only partially observed under this form of sampling. 

Our proposed procedures can be used in this setting. Let $\ind_i(d,t,\gamma) = \ind\{D_i=d, T_i=t, \gamma_i=\gamma\}$. The analog estimator for \eqref{tse} is
\begin{equation*}
  \frac{\sum_{i=1}^n Y_i \ind_i(d,t,\gamma)}{\sum_{i=1}^n \ind_i(d,t,\gamma)} - \frac{\sum_{i=1}^n Y_i \ind_i(d',t',\gamma)}{\sum_{i=1}^n \ind_i(d',t',\gamma)}. 
\end{equation*}

\noindent This fits into the setup of \eqref{ale2} by defining $Z_i = W_i$, 
\begin{equation*}
  \hat\theta = \left( \frac{1}{n} \sum_{i=1}^n \ind_i(d,t,\gamma), \frac{1}{n} \sum_{i=1}^n \ind_i(d',t',\gamma) \right), \quad\text{and}
\end{equation*}
\begin{equation*}
  \psi(Z_i; \beta_0, \hat\theta) = \frac{Y_i \ind_i(d,t,\gamma)}{n^{-1} \sum_{j=1}^n \ind_j(d,t,\gamma)} - \frac{Y_i \ind_i(d',t',\gamma)}{n^{-1} \sum_{j=1}^n \ind_j(d',t',\gamma)} - \beta_0, 
\end{equation*}

\noindent where $\beta_0$ is the hypothesized value of the true average treatment/spillover effect \eqref{tse}. 

\subsection{Network Statistics}\label{sns}

Inference methods for network statistics are important for network regressions, as discussed in \autoref{unknown}, and strategic models of network formation \citep{sheng2014}. They are also of inherent interest in the networks literature, where theoretical work is often motivated by ``stylized facts'' about the structure of real-world social networks \citep{barabasi2015,jackson2010}. These facts are obtained by computing various summary statistics from network data. However, little attempt has been made to account for the sampling variation of these point estimates, perhaps due to the wide variety of network formation models, which induce different dependence structures. This motivates the use of resampled statistics, which can be used to conduct inference on network statistics without taking a stance on the network formation model.

We next consider two stylized facts that have arguably received the most attention in the literature: clustering and power law degree distributions. This subsection focuses on the former, while the latter is discussed in a more general context in \autoref{stpl}. For a set of $n$ nodes, let $\bm{A}$ be a binary symmetric adjacency matrix that represents a network. Its $ij$th entry $A_{ij}$ is thus an indicator for whether $i$ and $j$ are linked. Define the {\em individual clustering} for a node $i$ in network $\bm{A}$ as
\begin{equation*}
  Cl_i(\bm{A}) = \frac{\sum_{j\neq i, k\neq j, k\neq i} A_{ij}A_{ik}A_{jk}}{\sum_{j\neq i, k\neq j, k\neq i} A_{ij}A_{ik}}, 
\end{equation*}

\noindent with $Cl_i(\bm{A}) \equiv 0$ if $i$ has at most one link. The denominator counts the number of pairs linked to $i$, while the numerator counts the number of {\em linked} pairs linked to $i$. The {\em average clustering coefficient} of $\bm{A}$ is defined as $n^{-1}\sum_{i=1}^n Cl_i(\bm{A})$.

This statistic is a common measure of {\em transitivity} or {\em clustering}, the tendency for individuals with partners in common to associate. A well-known stylized fact in the network literature is that most social networks exhibit nontrivial clustering, where ``nontrivial'' is defined relative to the null model in which links are i.i.d.\ \citep{jackson2010}. Under the null model, when $n$ is large, the average clustering coefficient is close to the probability of forming a link. Yet, the average clustering coefficient is typically larger than the empirical linking probability in practice, hence the stylized fact \citep[][Ch.\ 3]{barabasi2015}.

In order to formally assess whether average clustering is significantly different from the probability of link formation, we can use our tests \eqref{tests} with
\begin{equation*}
  X_i = Cl_i(\bm{A}) - \frac{2}{n-1} \sum_{j=1}^n A_{ij}. 
\end{equation*}

\noindent Then $\bar{X}$ is the difference between the average clustering coefficient and the empirical linking probability. To verify $\sqrt{n}$-consistency, we can apply CLTs derived by, for example, \cite{bickel_method_2011} and \cite{leung2017normal}.

\subsection{Testing for Power Laws}\label{stpl}

Testing whether data follows a power law distribution is of wide empirical interest in economics, finance, network science, neuroscience, biology, and physics \citep{barabasi2015,gabaix2009power,klaus2011statistical,newman2005power}. By ``power law'' we mean that the probability density or mass function of the data $f(x)$ is proportional to $x^{-\alpha}$ for some positive exponent $\alpha$. Many methods are available for estimating $\alpha$, for example maximum likelihood or regression estimators \citep{ibragimov2015heavy}. Given an estimate of the power law exponent, it is of interest to test how well the data accords with or deviates from a power law. Standard methods assume that the underlying data is i.i.d.\ \citep{clauset2009power}, but this is unrealistic for spatial, financial, and network data, motivating the use of resampled statistics. 

In the networks literature, a well-known stylized fact is that real-world social networks have power law degree distributions \citep{barabasi1999emergence}, where a node's degree is its number of connections. However, empirical evidence is commonly obtained by eyeballing log-log plots of degree distributions, rather than by using formal tests \citep{holme2018power}. \cite{broido2019scale} implement formal tests for power laws on a wide variety of network datasets, but their methods assume independent observations, despite the fact that network degrees are typically correlated.

The null hypothesis we next consider testing is motivated by \cite{klaus2011statistical}, which is that the power law fits no better than some reference null distribution, for example exponential or log-normal. This is operationalized using a Vuong test of the null that the expected log-likelihood ratio is zero. The numerator of the likelihood ratio is the power law distribution with an estimated exponent, and the denominator is the estimated null distribution. Under general misspecification, the log-likelihood ratio is zero if both models poorly fit the data, and less (greater) than zero if the null distribution fits better (worse) \citep{pesaran1987global,vuong1989likelihood}. 

For i.i.d.\ data and non-nested hypotheses, we can test the null by comparing the absolute value of the normalized log-likelihood ratio with a normal critical value. If the former is smaller, the models are equally good. Otherwise, we reject in favor of the power law (null distribution) if the log-likelihood ratio is greater (less) than zero.

We modify this procedure to account for dependence using the U-type statistic as follows. For identically distributed data $\{Z_i\}_{i=1}^n$, let $\ell_{PL}(Z_i,\alpha)$ be the likelihood of observation $i$ under a power law and $\ell_0(Z_i,\gamma)$ the likelihood under the null distribution, which is parameterized by $\gamma$. Then the null hypothesis is
\begin{equation*}
  \E\left[ \log \ell_{PL}(Z_i,\alpha) - \log \ell_0(Z_i,\gamma) \right] = 0. 
\end{equation*}

\noindent This fits into our setup \eqref{ale2} by defining $\hat\theta = (\hat\alpha, \hat\gamma)$ (the estimates of $(\alpha,\gamma)$) and 
\begin{equation*}
  X_i = \log \ell_{PL}(Z_i,\hat\alpha) - \log \ell_0(Z_i,\hat\gamma). 
\end{equation*}

\noindent We compute the U-type statistic using these $X_i$'s and compare it to a normal critical value. If the U-type statistic is smaller, then the models are equally good. Otherwise, we reject in favor of the power law (null distribution) if $\bar{X}$, the estimated log-likelihood ratio computed on the full dataset, is greater (less) than zero, as in \cite{vuong1989likelihood}. 

\section{Empirical Application}\label{sapp}

\cite{jackson2007meeting} propose a model of network formation that generates a degree distribution parameterized by $r$, which interpolates between the exponential and power law distributions. Their model provides microfoundations for the different distributions. When $r \rightarrow \infty$, the network is formed primarily through random meetings, and the distribution is exponential. When $r \rightarrow 0$, the network is formed primarily through ``network-based meetings,'' as nodes are more likely to meet neighbors of nodes that were previously met. Since high-degree nodes are more likely to be met through network-based meetings, this corresponds to a ``rich-get-richer'' or ``preferential-attachment'' mechanism that generates a power law degree distribution.

The authors estimate $r$ using data on six distinct social networks and informally assess the extent to which the estimated distributions depart from a power law. See their paper for descriptions of the data. In this section, we use the same datasets to implement the test described in \autoref{stpl}, using the exponential distribution as the null. We set the lower support points of the exponential and power law distributions at one and estimate the parameters of the distribution using (pseudo) maximum likelihood.  

\begin{table}[ht]
\centering
\caption{Results}
\begin{threeparttable}
\begin{tabular}{lrrrrrrrrrrrr}
\toprule
{} & \multicolumn{2}{c}{Coauthor} & \multicolumn{2}{c}{Radio} & \multicolumn{2}{c}{Prison} & \multicolumn{2}{c}{Romance} & \multicolumn{2}{c}{Citation} & \multicolumn{2}{c}{WWW} \\
\cmidrule{2-13}
Exp.\     &     \multicolumn{2}{c}{1.77} &     \multicolumn{2}{c}{1.46} &     \multicolumn{2}{c}{1.67} &     \multicolumn{2}{c}{1.94} &    \multicolumn{2}{c}{1.43} &      \multicolumn{2}{c}{1.79} \\
LL           &   \multicolumn{2}{c}{-36.45} &    \multicolumn{2}{c}{-1.75} &    \multicolumn{2}{c}{-6.13} &   \multicolumn{2}{c}{-13.75} &   \multicolumn{2}{c}{-3.84} &     \multicolumn{2}{c}{62.46} \\
Naive &         \multicolumn{2}{c}{E} &         \multicolumn{2}{c}{E} &         \multicolumn{2}{c}{E} &         \multicolumn{2}{c}{E} &        \multicolumn{2}{c}{E} &          \multicolumn{2}{c}{P} \\
$r$          &     \multicolumn{2}{c}{4.7} &     \multicolumn{2}{c}{5} &  \multicolumn{2}{c}{$\infty$} &  \multicolumn{2}{c}{$\infty$} &    \multicolumn{2}{c}{0.63} &      \multicolumn{2}{c}{0.57} \\
$n$          & \multicolumn{2}{c}{56639} &    \multicolumn{2}{c}{41} &    \multicolumn{2}{c}{60} &   \multicolumn{2}{c}{572} &  \multicolumn{2}{c}{233} & \multicolumn{2}{c}{325729} \\
\midrule
{} &      $R_n$ & RS &     $R_n$ & RS &   $R_n$ & RS &    $R_n$ & RS &    $R_n$ & RS &      $R_n$ & RS \\
\cmidrule{2-13}
 &  60k &  E &    37 &  N &  65 &  E & 1130 &  E &  692 &  E &  60k &  P \\
 &  80k &  E &    50 &  N &  86 &  E & 1507 &  E &  923 &  N &  80k &  P \\
RS & 100k &  E &    62 &  N & 108 &  E & 1884 &  E & 1154 &  N & 100k &  P \\
 & 120k &  E &    74 &  N & 130 &  E & 2261 &  E & 1385 &  E & 120k &  P \\
 & 140k &  E &    87 &  N & 151 &  E & 2638 &  E & 1616 &  E & 140k &  P \\
\bottomrule
\end{tabular}
\label{results}
\begin{tablenotes}[para,flushleft]
  \footnotesize ``Exp.'' $=$ estimated power law exponent. ``LL'' $=$ the normalized log-likelihood ratio. ``RS'' $=$ conclusion of our test, and ``Naive'' $=$ conclusion of i.i.d.\ test, where P $=$ power law, E $=$ exponential, and N $=$ fail to reject.
\end{tablenotes}
\end{threeparttable}
\end{table}

\autoref{results} displays the results of the tests. Row ``Exp.'' displays the estimated power law exponent, ``LL'' the normalized log-likelihood ratio, ``$r$'' the estimated value of $r$ from \cite{jackson2007meeting}, and $n$ the sample size. ``Naive'' displays the conclusion of the conventional Vuong test that assumes the data is i.i.d., with the conclusion ``P'' denoting power law, ``E'' denoting exponential, and ``N'' meaning the null is not rejected. Finally, the bottom five rows display the conclusions of our test for different values of $R_n$ to assess the robustness of the conclusions. Recalling the definition of $R_n^U$ from \eqref{R}, the rows correspond to $R_n = \min\{R_n^U \cdot \epsilon, 100\text{k}\}$ for $\epsilon \in \{0.6,0.8,1,1.2,1.4\}$. Thus, the middle of the bottom rows is our suggested choice $R_n^U$. We truncate $R_n$ at 100k since this is large enough to draw a robust conclusion.

The results of all three methods are in agreement for the coauthor, prison, romance, and WWW networks. This is due to the large values of the normalized log-likelihood ratios, which make the conclusion rather obvious regardless of the test used. For the other datasets, however, the methods draw different conclusions.

For the ham radio network, \cite{jackson2007meeting} estimate $r$ to be 5. As they note, this means network-based meetings are about eight times less common compared to the WWW network, so their degree distributions should be closer to the exponential than the power law distribution. However, our test finds insufficient evidence to reject the null for the ham radio network due to the very small sample size. The Vuong test rejects in favor of the exponential distribution, despite the normalized log-likelihood ratio being smaller (-1.75). This may be because the test assumes i.i.d.\ data, so the sample variance of the log-likelihoods may be underestimated.

\cite{jackson2007meeting} estimate $r$ to be close to zero for the citation network, which favors the power law. In contrast, the Vuong test rejects in favor of the exponential distribution, while our test either concludes exponential or fails to reject, depending on the value of $R_n$. This is due to the negative normalized log-likelihood -3.84, which is large enough for the i.i.d.\ test to draw a clear conclusion but perhaps not quite large after adjusting for dependence, which explains the ambiguity in the conclusion of our test. Still, neither test concludes the data is consistent with a power law, unlike the estimate of \cite{jackson2007meeting}.

\section{Large-Sample Theory}\label{smain}

This section considers a generalization of the setup in \autoref{smodel} in which $\bm{X}$ is a triangular array, so $X_i$ and $\mu_0$ may implicitly depend on $n$. This is important to accommodate network applications since, for example, when the network is sparse, the linking probability decays to zero with $n$. All proofs can be found in the appendix.

\subsection{CLT for Resampled Statistics}

For any vector $v$, let $\norm{v}$ denote its sup norm. The next theorem shows that the U-type (mean-type) statistic is asymptotically normal (chi-square).

\begin{theorem}\label{mq}
  For all $n$, let $R_n \geq 2$ be an integer. Suppose the following conditions hold under asymptotics sending $n\rightarrow\infty$.
  \begin{enumerate}[(a)]
    \item $n^{-1/2}\sum_{i=1}^n (X_i - \mu_0) = O_p(1)$. \label{consist0}
    \item There exists a positive-definite matrix $\bm{\Sigma}$ such that $\hat{\bm{\Sigma}} \plimarrow \bm{\Sigma}$. \label{sampvar0}
    \item $n^{-1} \sum_{i=1}^n \norm{X_i}^{2+\delta} = O_p(1)$ for some $\delta>0$. \label{4th0}
  \end{enumerate}

  \noindent If $R_n \rightarrow\infty$ and $R_n/n = o(1)$, then 
  \begin{equation*}
    T_M(\mu_0; \pi) \dlimarrow \chi^2_m \quad\text{and}\quad T_M(\bar{X}; \pi) \dlimarrow \chi^2_m, 
  \end{equation*}

  \noindent where $\chi^2_m$ is the chi-square distribution with $m$ degrees of freedom. If $R_n\rightarrow\infty$ and $\sqrt{R_n}/n = o(1)$, then
  \begin{equation*}
    T_U(\mu_0; \pi) \dlimarrow \mathcal{N}(0,1) \quad\text{and}\quad T_U(\bar{X}; \pi) \dlimarrow \mathcal{N}(0,1). 
  \end{equation*}

  \noindent Furthermore, these convergence statements are uniform in the sense that, conditional on $\bm{X}$, the cumulative distribution function (CDF) of each left-hand side statistic uniformly converges in probability to the CDF of the right-hand side distribution.
\end{theorem}

\begin{remark}\label{slower}
  While we focus on the standard case of $\sqrt{n}$-consistency, our methods can be adjusted to accommodate other rates of convergence, provided the rate is known. This involves choosing an asymptotically smaller value of $R_n$. To see this, suppose $\bar{X}$ is $n^\delta$-consistent for some $\delta \in (0,1)$, and consider decomposition \eqref{decomp} for the mean-type statistic. Term $[I]$ is still asymptotically normal. Term $[II]$ equals $\sqrt{R_n} \hat{\bm{\Sigma}}^{-1/2} n^{-1} \sum_{i=1}^n (X_i-\mu_0)$. Then given $\hat{\bm{\Sigma}} \plimarrow \bm{\Sigma}$ positive definite, we need
  \begin{equation*}
    \frac{\sqrt{R_n}}{n^\delta} \left( n^\delta \frac{1}{n} \sum_{i=1}^n (X_i-\mu_0) \right) = o_p(1).
  \end{equation*}

  \noindent The required rate condition is therefore $\sqrt{R_n}n^{-\delta} = o(1)$. 
\end{remark}

\begin{remark}
  \autoref{mq} provides limit distributions for $T_M(\bar{X}; \pi)$ and $T_U(\bar{X}; \pi)$, which yield an alternate way of constructing critical values. Define $\tilde\pi_\ell = (\tilde\pi_{r\ell})_{r=1}^{R_n}$ as in \autoref{nonrandom}. Following \S 3.2 of \cite{song_ordering-free_2016}, the ``permutation critical value'' for test \eqref{tests} is the $1-\alpha$ quantile of the ``permutation distribution,'' namely
  \begin{equation*}
    c_{L,1-\alpha} = \inf\left\{ c'>0\colon \frac{1}{L} \sum_{\ell=1}^L \ind\left\{ T_U(\bar{X};\tilde\pi_\ell) > c' \right\} \leq \alpha \right\} 
  \end{equation*}

  \noindent for the U-type statistic. Permutation critical values for the mean-type statistic are obtained analogously, replacing $T_U(\bar{X};\tilde\pi_\ell)$ with $T_M(\bar{X}; \tilde\pi_\ell)$ in the expression.
\end{remark}

\begin{remark}\label{power}
  The generality of our procedures comes at the cost of having power against fewer sequences of alternatives. If the econometrician could consistently estimate the asymptotic variance of $\bar{X}$, then the usual trinity of tests would have power against local alternatives $\mu_n = \mu_0 + h/\sqrt{n}$. In contrast, the test in \eqref{tests} using the mean-type statistic only has nontrivial asymptotic power against alternatives $\mu_n = h/\alpha_n$, where $\alpha_n \rightarrow \infty$ but $\alpha_n R_n^{-1/2} \rightarrow c \in [0,\infty)$. This is immediate from the rate of convergence and bias discussed in \autoref{smodel}. For the test using the U-type statistic, we have instead $\alpha_n R_n^{-1/4} \rightarrow c \in [0,\infty)$, following the argument in Theorem 3.3 of \cite{song_ordering-free_2016}. Due to the rate conditions on $R_n$, this implies that tests using our resampled statistics have lower power than conventional tests, which we interpret as the cost of dependence-robustness. Note that $R_n$ is chosen differently for the two statistics and can grow faster with $n$ for the U-type statistic, which is why it has better power properties.
\end{remark}

\begin{remark}\label{firststage}
  \autoref{mq} allows $X_i$ to depend on a ``first-stage'' estimator, which is important for many of the applications in \autoref{stapp}. Consider the setup for asymptotically linear estimators \eqref{ale2}. The following are primitive conditions for \autoref{mq}:
  \begin{enumerate}[(i)]
    \item $n^{-1/2} \sum_{i=1}^n \psi(Z_i;\beta_0,\theta_0)$ and $\sqrt{n}(\hat\theta - \theta_0)$ are $O_p(1)$.
    \item There exists $\hat\Sv$ consistent for $\Sv = \lim_{n\rightarrow\infty} n^{-1} \sum_{i=1}^n \E[\psi(Z_i;\beta_0,\theta_0) \psi(Z_i;\beta_0,\theta_0)']$, and the latter is positive-definite. 
    \item $n^{-1} \sum_{i=1}^n \norm{\psi(Z_i;\beta_0,\theta_0)}^{2+\delta} = O_p(1)$ for some $\delta>0$.
    \item $\sup_{\theta\in\Theta} \norm{n^{-1} \sum_{i=1}^n (\nabla_\theta \psi(Z_i;\beta_0,\theta) - \E[\nabla_\theta \psi(Z_i;\beta_0,\theta)])} = o_p(1)$.
  \end{enumerate}
\end{remark}

\subsection{Moment Inequality Test}\label{mils}

We next state formal results for test \eqref{mitest}. Let $\lambda_\text{min}(M)$ denote the smallest eigenvalue of a matrix $M$ and $\norm{M} = \max_{i,j} \abs{M_{ij}}$. Define $\mu_0(\prob) = \E_\prob[X_1]$, where $\E_\prob[\cdot]$ denotes the expectation under the data-generating process (DGP) $\prob$ (formally a probability measure). Let $\mu_{0k}(\prob)$ be the $k$th component of $\mu_0(\prob)$. 

\begin{theorem}\label{RSMIreal}
  Let $\mathcal{P}_0$ be the set of DGPs such that, for some $\delta,L,U \in (0,\infty)$ and any sequence $\{\prob_n\}_{n\in\mathbb{N}} \subseteq \mathcal{P}_0$, the following conditions hold.
  \begin{enumerate}[(a)]
    \item $n^{-1/2}\sum_{i=1}^n (X_i - \mu_0(\prob_n)) = O_{\prob_n}(1)$. 
    \item There exists $\bm{\Sigma}$ such that $\hat{\bm{\Sigma}} \plimarrow \bm{\Sigma}$, $\norm{\bm{\Sigma}}<U$, and $\lambda_\text{min}(\bm{\Sigma}) > L$.
    \item $n^{-1} \sum_{i=1}^n \norm{X_i}^{2+\delta} = O_{\prob_n}(1)$. 
    \item $\mu_0(\prob_n) \leq 0$ for all $n$ (null hypothesis).
  \end{enumerate}

  \noindent If $R_n\rightarrow\infty$ and $\sqrt{R_n}/n = o(1)$, then $\sup_{\prob\in\mathcal{P}_0} \E_\prob[\phi_n] \rightarrow \alpha$ for $\phi_n$ in \eqref{mitest}.
\end{theorem}

\noindent This shows that the test uniformly controls size. The theorem follows quite directly from the next lemma, which also provides results on the power of the test. 

\begin{lemma}\label{RSMI}
  Let $\{\prob_n\}_{n\in\mathbb{N}} \subseteq \mathcal{P}_0$ (defined in \autoref{RSMIreal}) and $\delta_k^* = \lim_{n\rightarrow \infty} R_n^{1/4} \mu_{0k}(\prob_n) \in \R \cup \{-\infty,\infty\}$. Suppose $R_n\rightarrow\infty$ and $\sqrt{R_n}/n = o(1)$.
  \begin{enumerate}[(a)]
    \item If $\max_k \delta_k^* \leq 0$ (null / ``local-to-null'' case), then $\E_{\prob_n}[\phi_n] \rightarrow \alpha$.
    \item If $\max_k \delta_k^* = \infty$ (fixed alternative case), then $\E_{\prob_n}[\phi_n] \rightarrow 1$.
    \item If $\max_k \delta_k^* \in (0,\infty)$ (local alternative case), then $\E_{\prob_n}[\phi_n] \rightarrow \beta > \alpha$.
  \end{enumerate}
\end{lemma}

\noindent Part (a) shows the test asymptotically controls size under any sequence of null DGPs. Parts (b) and (c) describe the test's power, with (c) showing that the test has power against local alternatives that vanish no faster than rate $R_n^{-1/4}$. Since size control requires $\sqrt{R_n}/n \rightarrow 0$, the test does not have power against $\sqrt{n}$ local alternatives.

\begin{remark}
  \autoref{power} discusses a sort of bias-variance trade-off for choosing $R_n$ in the equality-testing case. A similar trade-off occurs here since we use the same U-type statistic. Consider the case $m=1$ and drop the subscript $k$. By \eqref{r0u923jw} in the proof of \autoref{RSMI}, $\hat\lambda$ depends on the term $\sqrt{R_n}\hat{\bm{\Sigma}}^{-1/2}(\bar{X}-\mu_0(\prob_n))^2$. If $\sqrt{R_n}/n \rightarrow 0$ as required, then this term vanishes. However, if $R_n$ were too large, say equal to $n^2$, then the term would instead be asymptotically chi-square, and our test would have incorrect size. Validity of our test therefore requires $\sqrt{R_n}/n \rightarrow 0$ to eliminate a ``bias'' term. On the other hand, the rate of convergence of the test is $R_n^{-1/4}$, as shown in \autoref{RSMI}, which reflects the same trade-off for equality testing.
\end{remark}

\begin{remark}
  Consider the conventional moment inequalities setting in which $\bm{X}$ is i.i.d. For simplicity, suppose $m=1$, and consider the test statistic $n^{-1/2} \hat{\bm{\Sigma}}^{-1/2} \sum_{i=1}^n X_i$. The well-known difficulty with constructing critical values for this statistic is that while
  \begin{equation*}
    \frac{1}{\sqrt{n}} \sum_{i=1}^n \hat{\bm{\Sigma}}^{-1/2}(X_i-\mu_0(\prob_n)) \dlimarrow \mathcal{N}(\bm{0},\bm{I}_m), 
  \end{equation*}

  \noindent it is impossible to consistently estimate $\sqrt{n}\bm{\Sigma}^{-1/2}\mu_0(\prob_n)$. Much of the moment-inequalities literature boils down to finding clever ways to conservatively bound this nuisance parameter from above \citep{canay2018practical}. In contrast, in our setting, the nuisance parameter is \eqref{inputsplit}, whose leading term $\sqrt{R_n}(\bm{\Sigma}^{-1/2}\mu_0(\prob_n))^2$, for example, can be consistently estimated by $\sqrt{R_n}(\hat{\bm{\Sigma}}^{-1/2}\bar{X})^2$ since $\sqrt{R_n}/n\rightarrow 0$ \eqref{RXfact}. This is why our test is asymptotically exact.
\end{remark}

\section{Monte Carlo}\label{ssims}

This section presents results from four simulation studies, each corresponding to one of the applications in \autoref{smodel}. We use asymptotic critical values to implement the tests and multiple values of $R_n$ to assess the robustness of the methods to the tuning parameter. The results are broadly summarized as follows. The size is largely close to the target level of 5 percent across all designs with weak dependence, more so for larger sample sizes. Power can be low in small samples, as expected from the convergence rates discussed in \autoref{power}. The U-type statistic has significantly better power properties than the mean-type, which leads us to recommend the former. Finally, results are similar across values of $R_n$.

\bigskip

\noindent {\bf Cluster Dependence.} Let $c$ index cities, $f$ index families, and $i$ index individuals. We generate outcomes according to the random effects model
\begin{equation*}
  Y_{ifc} = \theta_0 + \alpha_f + \varepsilon_{ifc}, 
\end{equation*}

\noindent where $\alpha_f \stackrel{iid}\sim \mathcal{N}(0,1)$ and $\varepsilon_{ifc} \stackrel{iid}\sim \mathcal{N}(0,1)$, the two mutually independent. The correct level of clustering is at the family level, and the true value of $\theta_0$ is one. Let $n_c$, $n_f$, and $n_i$ be the number of cities, families, and individuals, respectively, and $N = (n_c, n_f, n_i)$. Families have equal numbers of individuals and cities equal numbers of families. 

We present results for resampled statistics and compare them to $t$-tests using clustered standard errors for each level of clustering. \autoref{CSE} displays simulation results for the size and power of our tests, computed using 6000 simulations. The first two rows display rejection percentages for size (testing $H_0\colon \theta_0=1$) and power (testing $H_0\colon \theta_0=1.5$), respectively. To show the robustness of our test, we display results for five different values of $R_n$. For the U-type statistic, the columns correspond to $R_n = R_n^U \cdot \epsilon$ for $\epsilon \in \{0.6,0.8,1,1.2,1.4\}$ and $R_n^U$ defined in \eqref{R}. Thus, the middle of the five columns for both sample sizes corresponds to our suggested choice $R_n^U$. We do the same for the M-type statistic, except we use $R_n = R_n^M\cdot\epsilon$.  Finally, \autoref{CSE2} displays analogous results for $t$-tests with clustered standard errors. The columns display the level of clustering with $c$ for city, $f$ for family, $i$ for individual.

The results show that the $t$-test overrejects when clustering at too coarse a level and the number of clusters is small (clustering at the city level). It also overrejects when clustering at too fine a level (clustering at the individual level) because this assumes more independence in the data than is warranted. In contrast, tests using resampled statistics properly control size. On the other hand, the $t$-test is clearly more powerful. U-type statistics show a significant power advantage over mean-type statistics. Our results are also similar across values of $R_n$.

\begin{table}[ht]
\centering
\caption{Weak Cluster Dependence: Our Tests}
\begin{threeparttable}
\resizebox{\columnwidth}{!}{ 
\begin{tabular}{llrrrrr|rrrrr}
\toprule
  & $N$   & \multicolumn{5}{c}{$(20,100,200)$} & \multicolumn{5}{c}{$(20,500,1000)$} \\
\cmidrule{2-12}
  & Size &          5.37 &   6.02 &   6.30 &   6.33 &   6.73 &           5.22 &    5.18 &    5.68 &    5.77 &    5.42 \\
M & Power &         17.55 &  22.62 &  27.88 &  32.52 &  36.70 &          33.35 &   44.35 &   51.63 &   58.77 &   65.78 \\
  & $R_n$ &          8 &  11 &  14 &  17 &  20 &          19 &   26 &   32 &   38 &   45 \\
\cmidrule{2-12}
  & Size &          5.68 &   6.63 &   6.10 &   6.32 &   6.30 &           5.53 &    5.17 &    5.32 &    5.58 &    5.67 \\
U & Power &         61.22 &  68.07 &  73.05 &  76.32 &  79.77 &          99.78 &   99.95 &   99.97 &  100 &  100 \\
  & $R_n$ &        278 & 371 & 464 & 557 & 650 &        2381 & 3175 & 3969 & 4763 & 5557 \\
\bottomrule
\end{tabular}}
\label{CSE}
\begin{tablenotes}[para,flushleft]
  \footnotesize Averages over 6000 simulations. $N = (n_c,n_f,n_i)$. $M =$ mean-type test, $U =$ U-type test.
\end{tablenotes}
\end{threeparttable}
\end{table}

\begin{table}[ht]
\centering
\caption{Weak Cluster Dependence: $t$-Tests}
\begin{threeparttable}
\begin{tabular}{lrrrrrr}
\toprule
$N$   & \multicolumn{3}{c}{$(20,100,200)$} & \multicolumn{3}{c}{$(20,500,1000)$} \\
\cmidrule{2-7}
Cluster Lvl &    $c$ &    $f$ &    $i$ &     $c$ &     $f$ &     $i$ \\
Size  &  6.93 &  5.27 & 11.37 &   7.02 &   5.18 &  11.22 \\
Power & 98.23 & 98.35 & 99.22 & 100 & 100 & 100 \\
\bottomrule
\end{tabular}
\label{CSE2}
\begin{tablenotes}[para,flushleft]
  \footnotesize Averages over 6000 simulations. $N = (n_c,n_f,n_i)$. 
\end{tablenotes}
\end{threeparttable}
\end{table}

These designs feature weak cluster dependence (required by all existing methods in the literature) since clustering is at the family level and the number of families is relatively large. To see how our methods break down under strong dependence, we next modify the design so that clustering is at the city level and the number of cities is small. For $N=(30,600,1200)$ (only 30 cities), the type I error of the mean-type statistic now ranges from 9.13 to 17.83 percent across the values of $R_n$, while that of the U-type statistic ranges from 25.70 to 33.58 percent. For $N=(10,600,1200)$, the error of the mean-type statistic ranges from 17.83 to 29.45 percent, and that of the U-type statistic ranges from 48.93 to 56.55 percent. 

\bigskip

\noindent {\bf Network Statistics.} We generate a network according to a strategic model of network formation, following the simulation design of \cite{leung2017}. There are $n$ nodes, and each node $i$ is endowed with a type $(X_i,Z_i)$, where $Z_i \stackrel{iid}\sim \text{Ber}(0.5)$ and $X_i \stackrel{iid}\sim U([0,1]^2)$, the two mutually independent. Let $\rho$ be the function such that $\rho(\delta)=0$ if $\delta \leq 1$ and equal to $\infty$ otherwise. Potential links satisfy
\begin{equation*}
  A_{ij} = \ind\left\{ \theta_1 + (Z_i+Z_j) \theta_2 + \max_k G_{ik}G_{jk} \theta_3 - \rho(r_n^{-1}\norm{X_i-X_j}) + \zeta_{ij} > 0 \right\}, 
\end{equation*}

\noindent where $\norm{\cdot}$ is the Euclidean norm on $\R^d$ and $\zeta_{ij} \stackrel{iid}\sim \mathcal{N}(0,\theta_4^2)$ is independent of types. We set $\theta = (-1, 0.25, 0.25, 1)$ and $r_n = (3.6/n)^{1/2}$ and use the selection mechanism in the design of \cite{leung2017}; see his paper for details.

We are interested in two statistics that are functions of the network, the average clustering coefficient (defined in \autoref{sns}) and the average degree. \cite{leung2017normal} prove $\sqrt{n}$-consistency of the sample statistics for their population analogs. Let $\theta_0$ be the expected value of the network statistic. Tables \ref{tnetstatscc} and \ref{tnetstatsdeg} in the appendix display rejection percentages for average clustering and degree for two nulls. The first is that $\theta_0$ equals its true value, which estimates the size. The second is that $\theta_0$ equals the true value plus the number indicated in the table, which estimates power. We use 6000 simulation draws each to simulate $\theta_0$ and rejection percentages.

For this design, our tests exhibit substantial size distortion in small samples ($n=100$), but the size tends toward the nominal level as $n$ grows. For the average degree, there is still some size distortion at larger samples ($n=1000$) for the U-type statistic, although this is less than 2 percentage points above the nominal level. The U-type statistic is significantly more powerful than the mean-type statistic, with rejection percentages sometimes more than twice as large.

\bigskip

\noindent {\bf Treatment Spillovers.} Consider the setup in \autoref{stes}. We assign units to treatment with probability 0.3 and draw the network from the same model used for the network statistics above. The outcome model is $Y_i = \beta_1 + \beta_2 D_i + \beta_3 T_i + \beta_4 \gamma_i + \varepsilon_i$. For $\nu_j \stackrel{iid}\sim \mathcal{N}(0,1)$ independent of treatments, we set $\varepsilon_i = \nu_i + \sum_j A_{ij}\nu_j/\sum_j A_{ij}$, which represents exogenous peer effects in unobservables and generates network autocorrelation in the errors. We set $(\beta_1,\beta_2,\beta_3,\beta_4) = (1, 0.5, -1, 0.5)$.

We consider a linear regression estimator of $Y_i$ on $(1, D_i, T_i, \gamma_i)$ and test two hypotheses, $\beta_3=-1$ to estimate the size and $\beta_3=-1.8$ to estimate the power. \autoref{tspill} in the appendix displays rejection percentages computed using 6000 simulations. As in the network statistics application, we display multiple values of $R_n$ corresponding to $R_n = R_n^U\cdot \epsilon$ for the U-type statistic and $R_n = R_n^M\cdot\epsilon$ for the mean-type statistics, for the same values of $\epsilon$ above. In this design, our tests control size well across all sample sizes, and the U-type statistic is substantially more powerful.

\bigskip

\noindent {\bf Power Laws.} We implement the test in \autoref{stpl}, which uses the U-type statistic. We set $R_n = R_n^U$. Following the notation in that section, we draw data $\{Z_i\}_{i=1}^n$ i.i.d.\ from either an $\text{Exp}(0.5)$ distribution or a power law distribution with exponent 2. The lower support point for both is set at 1. \autoref{tPL} in the appendix reports rejection percentages from 6000 simulations under both alternatives (exponential and power law). Row ``LL'' displays the average normalized log-likelihood ratio, ``Favor Exp'' the percentage of simulations in which we reject in favor of the null distribution, and ``Favor PL'' the percentage of simulations in which we reject in favor of the power law. The power is around 55--60 percent for $n=100$ and 86--97 percent for $n=500$.

\section{Conclusion}\label{sconclude}

We develop tests for moment equalities and inequalities that are robust to general forms of weak dependence. The tests compare a resampled test statistic to an asymptotic critical value, in contrast to conventional resampling procedures, which compare a test statistic constructed using the original dataset to a resampled critical value. The validity of conventional procedures requires resampling in a way that mimics the dependence structure of the data, which in turn requires knowledge about the type of dependence. In contrast, the procedures we study are implemented the same way regardless of the dependence structure. We show our methods are asymptotically valid under the weak requirement that the target parameter can be estimated at a $\sqrt{n}$ rate. To illustrate the broad applicability of our procedure, we discuss four applications, including regression with unknown dependence, treatment effects with network interference, and testing network stylized facts.


\FloatBarrier
\phantomsection
\addcontentsline{toc}{section}{References}
\bibliography{dependence_robust}{} 
\bibliographystyle{aer}

\appendix
\numberwithin{equation}{section} 
\numberwithin{table}{section}

\section{Proofs}\label{sproofs}

\begin{proof}[Proof of \autoref{mq}]
  By Theorems \ref{permgen} and \ref{quadstat} the CDFs of the test statistics converge in probability, pointwise, to those of their stated limits conditional on $\bm{X}$. By an extension of Poly\'{a}'s Theorem \citep[e.g.][Theorem 11.2.9]{lehmann2006testing}, the CDFs converge uniformly. Furthermore, the bounded convergence theorem implies unconditional convergence.
\end{proof}

\begin{theoremSA}[Mean-Type Statistic]\label{permgen}
  For all $n$, let $R_n\geq 1$ be an integer. Suppose the following conditions hold.
  \begin{enumerate}[(a)]
    \item $R_n\rightarrow\infty$ and $R_n/n = o(1)$. \label{rates}
    \item $n^{-1/2}\sum_{i=1}^n (X_i - \mu_0) = O_p(1)$. \label{consist}
    \item There exists a positive-definite matrix $\bm{\Sigma}$ such that $\hat{\bm{\Sigma}} \plimarrow \bm{\Sigma}$. \label{sampvar}
    \item $n^{-1} \sum_{i=1}^n \norm{X_i}^{2+\delta} = O_p(1)$ for some $\delta>0$. \label{4th}
  \end{enumerate}

  \noindent Then
  \begin{equation*}
    \tilde T_M(\mu_0; \pi) \dlimarrow \mathcal{N}(\bm{0},\bm{I}_m) \quad\text{and}\quad \tilde T_M(\bar{X}; \pi) \dlimarrow \mathcal{N}(\bm{0},\bm{I}_m).\footnote{For this theorem and \autoref{quadstat}, by ``$\dlimarrow$'' we mean the CDFs of the statistics converge in probability, pointwise, to those of their stated limits conditional on $\bm{X}$.}
  \end{equation*}
\end{theoremSA}
\begin{proof}
  Recall decomposition \eqref{decomp}. Let $\tilde{X}_i = X_i - \mu_0$. By definition of $\pi_r$,
  \begin{align*}
    \E[\tilde T_M(\mu_0; \pi) \mid \bm{X}] &= \sqrt{R_n}\hat{\bm{\Sigma}}^{-1/2} \frac{1}{\abs{\Pi}} \sum_{\pi\in\Pi} \tilde{X}_{\pi(1)} = \sqrt{R_n}\hat{\bm{\Sigma}}^{-1/2} \frac{1}{n!} \sum_{\pi\in\Pi} \tilde{X}_{\pi(1)} \\
    &= \sqrt{R_n}\hat{\bm{\Sigma}}^{-1/2} \frac{1}{n!} \sum_{i=1}^n \tilde{X}_i (n-1)! = \sqrt{\frac{R_n}{n}} \hat{\bm{\Sigma}}^{-1/2} \frac{1}{\sqrt{n}} \sum_{i=1}^n \tilde{X}_i. 
  \end{align*}

  \noindent Hence, $[II]$ in \eqref{decomp} is $O_p( (R_n/n)^{1/2} )$ and therefore $o_p(1)$ by our assumptions.

  We next condition throughout on $\bm{X}$ and apply a Lindeberg CLT to $[I]$ in \eqref{decomp}, which is a sum of conditionally independent random vectors. First we show that asymptotic variance is $\bm{I}_m$. We have
  \begin{equation*}
    W_{M,r} - \E[W_{M,r} \mid \bm{X}] = \hat{\bm{\Sigma}}^{-1/2} \tilde{X}_{\pi_r(1)} - \frac{1}{n} \sum_{i=1}^n \hat{\bm{\Sigma}}^{-1/2} \tilde{X}_i = \hat{\bm{\Sigma}}^{-1/2} (X_{\pi_r(1)} - \bar{X}). 
  \end{equation*}

  \noindent Its conditional second moment is
  \begin{multline*}
    \frac{1}{\abs{\Pi}} \sum_{\pi\in\Pi} \hat{\bm{\Sigma}}^{-1/2} (X_{\pi(1)}-\bar{X}) (X_{\pi(1)} - \bar{X})' (\hat{\bm{\Sigma}}^{-1/2})' \\
    = \frac{1}{n!} \sum_{i=1}^n \hat{\bm{\Sigma}}^{-1/2}(X_i-\bar{X})(X_i-\bar{X})' (\hat{\bm{\Sigma}}^{-1/2})' (n-1)! = \bm{I}_m.
  \end{multline*}

  \noindent Since $\var( W_{M,r} \mid \bm{X} )$ is identically distributed across $r$, $\var([I] \mid \bm{X}) = \bm{I}_m$.

  Similar calculations yield, for $\delta$ in assumption \eqref{4th},
  \begin{equation*}
    \E\left[ \norm{W_{M,r} - \E[W_{M,r} \mid \bm{X}]}^{2+\delta} \mid \bm{X} \right] = \frac{1}{n} \sum_{i=1}^n \norm{\hat{\bm{\Sigma}}^{-1/2}(X_i - \mu_0) - \hat{\bm{\Sigma}}^{-1/2}(\bar{X}-\mu_0)}^{2+\delta}, 
  \end{equation*}

  \noindent where $\norm{\cdot}$ denotes the sup norm for vectors. This is $O_p(1)$ by assumptions \eqref{consist}--\eqref{4th} and Minkowski's inequality, which verifies the Lindeberg condition.
\end{proof}

\begin{theoremSA}[U-Type Statistic]\label{quadstat}
  For all $n$, let $R_n \geq 2$ be an integer. If (a) $R_n\rightarrow\infty$, $\sqrt{R_n}/n = o(1)$, and (b)--(d) assumptions \eqref{consist}--\eqref{4th} of \autoref{permgen} hold, then
  \begin{equation*}
    T_U(\mu_0; \pi) \dlimarrow \mathcal{N}(0,1) \quad\text{and}\quad T_U(\bar{X}; \pi) \dlimarrow \mathcal{N}(0,1). 
  \end{equation*}
\end{theoremSA} 
\begin{proof}
  {\bf Step 1.} We first show that $T_U(\bar{X}; \pi) = T_U(\mu_0; \pi) + o_p(1)$. We have
  \begin{multline*}
    T_U(\bar{X}; \pi) = T_U(\mu_0; \pi) - (\bar{X}-\mu_0)'\hat{\bm{\Sigma}}^{-1} \frac{1}{\sqrt{mR_n}} \sum_{r=1}^{R_n} \left( (X_{\pi_r(1)}-\mu_0) + (X_{\pi_r(2)}-\mu_0) \right) \\ + \frac{\sqrt{R_n/m}}{n} \sqrt{n}(\bar{X}-\mu_0)'\hat{\bm{\Sigma}}^{-1} \sqrt{n}(\bar{X}-\mu_0). 
  \end{multline*}

  \noindent By assumption \eqref{consist}, the third term on the right-hand side is $o_p(1)$. The second term on the right-hand equals $-1$ times the sum of two similar terms, one of which is 
  \begin{equation*}
    (\bar{X}-\mu_0)'\hat{\bm{\Sigma}}^{-1} \underbrace{\frac{1}{\sqrt{mR_n}} \sum_{r=1}^{R_n} (X_{\pi_r(1)}-\mu_0)}_{A_n}. 
  \end{equation*}

  \noindent As shown in the proof of \autoref{permgen}, $A_n = O_p(1)$. Thus, the previous expression is $o_p(1)$ by assumption \eqref{consist}.

  \bigskip

  \noindent {\bf Step 2.} Decompose
  \begin{equation}
    T_U(\mu_0; \pi) = (T_U(\mu_0; \pi) - \E[T_U(\mu_0; \pi) \mid \bm{X}]) + \E[T_U(\mu_0; \pi) \mid \bm{X}]. \label{decompose}
  \end{equation}

  \noindent We show that $\E[T_U(\mu_0; \pi) \mid \bm{X}] \plimarrow 0$:
  \begin{align*}
    \E[T_U(\mu_0; \pi) \mid \bm{X}] &= \frac{1}{\sqrt{mR_n}} \sum_{r=1}^{R_n} \E\left[ (X_{\pi_r(1)}-\mu_0)' \hat{\bm{\Sigma}}^{-1} (X_{\pi_r(2)}-\mu_0) \,\big|\, \bm{X} \right] \\
				    &= \sqrt{\frac{R_n}{m}} \frac{1}{\abs{\Pi}} \sum_{\pi\in\Pi} (X_{\pi(1)}-\mu_0)' \hat{\bm{\Sigma}}^{-1} (X_{\pi(2)}-\mu_0) \\
    &= \sqrt{\frac{R_n}{m}} \frac{1}{n(n-1)} \sum_{i=1}^n \sum_{j\neq i} (X_i-\mu_0)' \hat{\bm{\Sigma}}^{-1} (X_j-\mu_0).
  \end{align*}

  \noindent From the last line, add and subtract
  \begin{equation*}
    \sqrt{\frac{R_n}{m}} \frac{1}{n(n-1)} \sum_{i=1}^n (X_i-\mu_0)' \hat{\bm{\Sigma}}^{-1} (X_i-\mu_0) 
  \end{equation*}

  \noindent to obtain
  \begin{equation*}
    \frac{\sqrt{R_n/m}}{n-1} \left( \sqrt{n}(\bar{X}-\mu_0)' \hat{\bm{\Sigma}}^{-1} \sqrt{n}(\bar{X}-\mu_0) - \frac{1}{n} \sum_{i=1}^n (X_i-\mu_0)' \hat{\bm{\Sigma}}^{-1} (X_i-\mu_0) \right). 
  \end{equation*}

  \noindent This is $o_p(1)$ by assumptions (a)--(c).

  \bigskip

  \noindent {\bf Step 3.} It remains to establish a normal limit for the term $T_U(\mu_0; \pi) - \E[T_U(\mu_0; \pi) \mid \bm{X}]$ in decomposition \eqref{decompose}. We condition on the data, treating it as fixed, and apply a Lindeberg CLT, noting $T_U(\mu_0; \pi)$ is an average of conditionally independent variables
  \begin{equation*}
    W_{U,r} = m^{-1/2} (X_{\pi_r(1)}-\mu_0)' \hat{\bm{\Sigma}}^{-1} (X_{\pi_r(2)}-\mu_0).
  \end{equation*}
  
  First consider the variance. We have $\var(T_U(\mu_0; \pi) \mid \bm{X}) = \E[W_{U,r}^2 \mid \bm{X}] - \E[W_{U,r} \mid \bm{X}]^2$, where the second term on the right-hand side is $o_p(1)$ by step 2 and 
  \begin{align*}
    \E[W_{U,r}^2 \mid \bm{X}] &= \frac{1}{\abs{\Pi}} \sum_{\pi\in\Pi} m^{-1} \big( (X_{\pi_r(1)}-\mu_0)' \hat{\bm{\Sigma}}^{-1} (X_{\pi_r(2)}-\mu_0) \big)^2 \\
    &= \frac{1}{n(n-1)} \sum_{i=1}^n \sum_{j\neq i} m^{-1} \big( (X_i-\mu_0)' \hat{\bm{\Sigma}}^{-1} (X_j-\mu_0) \big)^2 \\
    &= \frac{1}{n} \sum_{i=1}^n m^{-1} (X_i-\mu_0)' \hat{\bm{\Sigma}}^{-1} \left( \frac{1}{n-1} \sum_{j\neq i} (X_j-\mu_0) (X_j-\mu_0)'\right) \hat{\bm{\Sigma}}^{-1} (X_i-\mu_0), 
  \end{align*}

  \noindent which converges in probability to one, as desired.
  
  Finally, we show that, for $\delta$ in assumption \eqref{4th},
  \begin{equation*}
    \E\left[\abs{W_{U,r}-\E[W_{U,r} \mid \bm{X}]}^{2+\delta} \,\bigg|\, \bm{X}\right] = O_p(1). 
  \end{equation*}

  \noindent This is enough to verify the Lindeberg condition since $W_{U,r}$ is identically distributed across $r$. The left-hand side of the previous equation is equal to a constant times
  \begin{multline*}
    \frac{1}{n(n-1)} \sum_{i=1}^n \sum_{j\neq i} \abs{(X_i-\mu_0)'\hat{\bm{\Sigma}}^{-1}(X_j-\mu_0) - \E[W_{U,r} \mid \bm{X}]}^{2+\delta} \\ \leq \left( \left( \frac{1}{n(n-1)} \sum_{i=1}^n \sum_{j\neq i} \abs{(X_i-\mu_0)'\hat{\bm{\Sigma}}^{-1}(X_j-\mu_0)}^{2+\delta} \right)^{1/(2+\delta)} + \abs{\E[W_{U,r} \mid \bm{X}]} \right)^{2+\delta}
  \end{multline*}

  \noindent by Minkowski's inequality. This is $O_p(1)$ since
  \begin{equation*}
    \frac{1}{n(n-1)} \sum_{i=1}^n \sum_{j\neq i} \abs{(X_i-\mu_0)'\hat{\bm{\Sigma}}^{-1}(X_j-\mu_0)}^{2+\delta} = O_p(1)
  \end{equation*}

  \noindent by assumptions \eqref{sampvar} and \eqref{4th}.
\end{proof}

\vspace{15pt}
\begin{proof}[Proof of \autoref{RSMIreal}]
  We prove $\sup_{\prob\in\mathcal{P}_0} \E_\prob[\phi_n] \rightarrow \alpha' \leq \alpha$ by contradiction. Suppose not. Then we can find some null sequence $\{\prob_n\}_{n\in\mathbb{N}} \subseteq \mathcal{P}_0$ such that $\liminf_{n\rightarrow\infty} \E_{\prob_n}[\phi_n] > \alpha$. This contradicts conclusion (a) of \autoref{RSMI}. 
  
  Since $\mathcal{P}_0$ includes a DGP $\prob$ under which $\E_\prob[X_1] \leq 0$, setting $\prob_n = \prob$ for all $n$ yields $\E_\prob[\phi_n] \rightarrow \alpha$ by conclusion (a) of \autoref{RSMI}. Hence, $\alpha'=\alpha$.
\end{proof}

\vspace{15pt}
\begin{proof}[Proof of \autoref{RSMI}]
  We first establish the asymptotic behavior of $\hat\lambda_k$. Note that 
  \begin{align}
    m^{1/2} \hat\lambda_k =& (\bar{X}_k-\mu_{0k}(\prob_n)) \hat{\bm{\Sigma}}_{kk}^{-1} \frac{1}{\sqrt{R_n}} \sum_{r=1}^{R_n} \left( (X_{\pi_r(1),k}-\mu_{0k}(\prob_n)) + (X_{\pi_r(2),k}-\mu_{0k}(\prob_n)) \right) \nonumber\\
		   &- \sqrt{R_n} \hat{\bm{\Sigma}}_{kk}^{-1} (\bar{X}_k-\mu_{0k}(\prob_n))^2 \nonumber\\
		   &+ \mu_{0k}(\prob_n) \hat{\bm{\Sigma}}_{kk}^{-1} \frac{1}{\sqrt{R_n}} \sum_{r=1}^{R_n} \left( (X_{\pi_r(1),k}-\mu_{0k}(\prob_n)) + (X_{\pi_r(2),k}-\mu_{0k}(\prob_n)) \right) \nonumber\\
		   &+ 2 \bar{X}_k \hat{\bm{\Sigma}}_{kk}^{-1} \sqrt{R_n} \mu_{0k}(\prob_n) + \sqrt{R_n} \hat{\bm{\Sigma}}_{kk}^{-1} (\mu_{0k}(\prob_n)^2 - 2\bar{X}_k\mu_{0k}(\prob_n)). \label{r0u923jw}
  \end{align}

  \noindent As shown in the proof of \autoref{permgen}, 
  \begin{equation*}
    \frac{1}{\sqrt{R_n}} \sum_{r=1}^{R_n} \left( (X_{\pi_r(1)}-\mu_0(\prob_n)) + (X_{\pi_r(2)}-\mu_0(\prob_n)) \right) = O_{\prob_n}(1).
  \end{equation*}
  
  \noindent Since $\hat{\bm{\Sigma}}_{kk}$ is asymptotically bounded away from zero and infinity under our assumptions, the first and second lines on the right-hand side of \eqref{r0u923jw} are $o_{\prob_n}(1)$; the third line is $\mu_{0k}(\prob_n) \cdot O_{\prob_n}(1)$; and the last line equals $\sqrt{R_n} \hat{\bm{\Sigma}}_{kk}^{-1} \mu_{0k}(\prob_n)^2$. Then
  \begin{equation}
    \hat\lambda_k = \sqrt{\frac{R_n}{m}} \hat{\bm{\Sigma}}_{kk}^{-1} \mu_{0k}(\prob_n)^2 + \mu_{0k}(\prob_n) \cdot O_{\prob_n}(1) + o_{\prob_n}(1). \label{inputsplit}
  \end{equation}

  \noindent We also note for later that
  \begin{equation}
    R_n^{1/4}\bar{X}_k = R_n^{1/4}\mu_{0k}(\prob_n) + \sqrt{\frac{R_n^{1/2}}{n}} \sqrt{n}(\bar{X}_k - \mu_{0k}(\prob_n)) = R_n^{1/4}\mu_{0k}(\prob_n) + o_{\prob_n}(1). \label{RXfact}
  \end{equation}

  \noindent Now we turn to each of the claims (a)--(c) of the lemma.

  \bigskip

  \noindent {\bf Claim (a).} Recall that $\delta_k^* = \lim_{n\rightarrow\infty} R_n^{1/4} \mu_{0k}(\prob_n)$. Suppose $\max_k \delta_k^* \leq 0$. Without loss of generality, assume $\delta_k^* = 0$ for all $k=1, \dots, \ell$, and $\delta_k^* < 0$ for all $k=\ell+1, \dots, m$. Then for any $k=1, \dots, \ell$,
  \begin{equation}
    \abs{\hat\lambda_k \ind\{R_n^{1/4}\bar{X}_k \geq 0\}} \leq \abs{\hat\lambda_k} \plimarrow 0 \label{fgw903j4}
  \end{equation}

  \noindent by \eqref{inputsplit}. For any $k=\ell+1, \dots, m$, we have $\hat\lambda_k \ind\{R_n^{1/4}\bar{X}_k \geq 0\} \plimarrow 0$ because, for any $\epsilon > 0$, by the law of total probability,
  \begin{multline*}
    \prob_n\left( \abs{\hat\lambda_k \ind\{R_n^{1/4}\bar{X}_k \geq 0\}} > \epsilon \right) \leq \underbrace{\prob_n\left( \abs{\hat\lambda_k \ind\{R_n^{1/4}\bar{X}_k \geq 0\}} > \epsilon \medcap R_n^{1/4}\bar{X}_k < 0 \right)}_0 \\ + \prob_n(R_n^{1/4}\bar{X}_k \geq 0) \rightarrow 0
  \end{multline*}

  \noindent by \eqref{RXfact} since $\delta_k^* < 0$. Therefore, by \eqref{algebra} and the requirement $R_n>0$,
  \begin{equation}
    \prob_n(Q_n(\pi) > c_{1-\alpha}) = \prob_n\left( \max_{1\leq k\leq m} \left\{ T_{U,k}(\bar{X}_k;\pi) + o_{\prob_n}(1) \right\} > c_{1-\alpha} \right). \label{hof09hb3n}
  \end{equation}

  \noindent By definition, $c_{1-\alpha}$ is the $(1-\alpha)$-quantile of $\tilde Q_n(\pi)$. Moreover, $\tilde Q_n(\pi)$ is a continuous function of $(T_{U,k}(\bar{X}_k;\pi))_{k=1}^m$, whose CDF converges uniformly to that of a normal random vector by a minor extension of \autoref{quadstat}. Hence,
  \begin{equation*}
    \eqref{hof09hb3n} = \prob_n(\tilde Q_n(\pi) + o_{\prob_n}(1) > c_{1-\alpha}) \rightarrow \alpha.
  \end{equation*}
  
  \bigskip
  
  \noindent {\bf Claim (b).} Suppose $\max_k \delta_k^* = \infty$. Then for some $k = 1,\dots, m$, $\hat\lambda_k\ind\{\bar{X}_k\geq 0\} \plimarrow \infty$ by \eqref{inputsplit}. Since $(T_{U,k}(\bar{X}_k;\pi))_{k=1}^m$ has a tight limit distribution, as established in claim (a), we have $Q_n(\pi) \plimarrow \infty$ by \eqref{algebra}. On the other hand, $c_{1-\alpha}$ converges in probability to a positive constant. Hence, the rejection probability tends to one.

  \bigskip
  
  \noindent {\bf Claim (c).} Suppose $\max_k \delta_k^* \in (0,\infty)$. Without loss of generality suppose that $\delta_k^*$ is finite and strictly positive for $k=1, \dots, \ell$ and non-positive for $k=\ell+1, \dots, m$. Then for $k=1, \dots, \ell$, by \eqref{inputsplit} and \eqref{RXfact},
  \begin{equation*}
    \hat\lambda_k \ind\{R_n^{1/4}\bar{X}_k \geq 0\} \plimarrow \lambda_k^* \equiv \frac{(\delta_k^*)^2 \bm{\Sigma}_{kk}^{-1}}{\sqrt{m}} \in (0,\infty),
  \end{equation*}
  
  \noindent where $\bm{\Sigma}_{kk}$ is the $k$th diagonal of $\bm{\Sigma}$. For $k=\ell+1, \dots, m$, $\hat\lambda_k \ind\{R_n^{1/4}\bar{X}_k \geq 0\} \plimarrow 0$, as shown in claim (a). Therefore, by \eqref{algebra},
  \begin{equation}
    Q_n(\pi) = \max\left\{ \max_{1\leq k\leq\ell} \left\{ T_{U,k}(\bar{X}_k;\pi) + \lambda_k^* + o_{\prob_n}(1) \right\}, \max_{\ell < k \leq m} \left\{ T_{U,k}(\bar{X}_k;\pi) + o_{\prob_n}(1) \right\} \right\}. \label{fwj2}
  \end{equation}

  \noindent On the other hand, $c_{1-\alpha}$ is the $(1-\alpha)$-quantile of the distribution of
  \begin{equation*}
    \tilde Q_n(\pi) = \max\left\{ \max_{1\leq k\leq\ell} T_{U,k}(\bar{X}_k;\pi), \max_{\ell < k \leq m} T_{U,k}(\bar{X}_k;\pi) \right\}. 
  \end{equation*}

  \noindent As discussed above in claim (a), both have tight limit distributions obtained by replacing $(T_{U,k}(\bar{X}_k;\pi))_{k=1}^m$ with a normal random vector. Since the $\lambda_k^*$s in \eqref{fwj2} are strictly positive, $\prob_n(Q_n(\pi) > c_{1-\alpha})$ converges to some $\beta > \alpha$.
\end{proof}

\section{Additional Tables}

This section contains simulation results referenced in \autoref{ssims}.

\begin{table}[ht]
\centering
\caption{Average Clustering}
\resizebox{\columnwidth}{!}{
\begin{threeparttable}
\begin{tabular}{llrrrrr|rrrrr}
\toprule
$n$   &       & \multicolumn{5}{c}{Mean-Type Test} & \multicolumn{5}{c}{U-Type Test} \\
\cmidrule{2-12}
    & Size &           5.35 &  6.67 &  7.58 &  7.50 &  8.93 &        8.88 &    9.75 &   10.08 &   10.47 &   10.57 \\
100 & Power &          13.80 & 17.90 & 20.63 & 23.38 & 26.05 &       32.40 &   35.75 &   38.67 &   39.85 &   41.68 \\
   & $R_n$ &           6 &  8 & 10 & 12 & 14 &      110 &  147 &  184 &  221 &  258 \\
\cmidrule{2-12}
 & Size &           4.77 &  5.57 &  5.62 &  6.10 &  6.12 &        6.77 &    6.90 &    7.03 &    7.57 &    7.68 \\
500 & Power &          19.45 & 26.03 & 30.42 & 33.42 & 39.32 &       69.75 &   74.87 &   79 &   81.98 &   85.08 \\
   & $R_n$ &          13 & 18 & 22 & 26 & 31 &      944 & 1259 & 1574 & 1889 & 2204 \\
\cmidrule{2-12}
 & Size &           4.93 &  5.32 &  5.92 &  5.30 &  6.00 &        5.85 &    6.02 &    6.65 &    6.10 &    6.32 \\
1k & Power &          25.73 & 32.97 & 38.12 & 44.50 & 47.92 &       91.87 &   95.02 &   96.55 &   97.48 &   98.12 \\
   & $R_n$ &          19 & 25 & 31 & 37 & 43 &     2381 & 3174 & 3968 & 4762 & 5555 \\
\bottomrule
\end{tabular}
\label{tnetstatscc}
\begin{tablenotes}[para,flushleft]
  \footnotesize Averages over 6000 simulations. ``Size'' rows obtained from testing $H_0\colon \theta_0=\theta^*$, where $\theta^*=$ true expected value of the average clustering. ``Power'' rows obtained from testing $H_0\colon \theta_0=\theta^*+0.08$.
\end{tablenotes}
\end{threeparttable}}
\end{table}

\begin{table}[ht]
\centering
\caption{Average Degree}
\resizebox{\columnwidth}{!}{
\begin{threeparttable}
\begin{tabular}{llrrrrr|rrrrr}
\toprule
$n$   &       & \multicolumn{5}{c}{Mean-Type Test} & \multicolumn{5}{c}{U-Type Test} \\
\cmidrule{2-12}
 & Size &           7.22 &  7.43 &  8.18 &  9.10 &  9.77 &        9.20 &   10.22 &   10.83 &   10.88 &   12.07 \\
100 & Power &          26.58 & 31.98 & 36.72 & 41.08 & 45.93 &       62.28 &   67.07 &   69.87 &   72.50 &   74.45 \\
   & $R_n$ &           6 &  8 & 10 & 12 & 14 &      110 &  147 &  184 &  221 &  258 \\
\cmidrule{2-12}
 & Size &           5.77 &  5.97 &  6.07 &  6.55 &  7.25 &        6.68 &    7.18 &    8.08 &    8.30 &    8.23 \\
500 & Power &          36.62 & 49.27 & 55.10 & 60.72 & 67.95 &       96.88 &   98.43 &   98.98 &   99.37 &   99.25 \\
   & $R_n$ &          13 & 18 & 22 & 26 & 31 &      944 & 1259 & 1574 & 1889 & 2204 \\
\cmidrule{2-12}
 & Size &           5.65 &  5.92 &  5.40 &  5.85 &  6.07 &        6.05 &    6.43 &    6.82 &    6.63 &    6.75 \\
1k & Power &          49.17 & 59.52 & 68.52 & 75.10 & 80.45 &       99.93 &   99.97 &   99.98 &   99.98 &  100 \\
   & $R_n$ &          19 & 25 & 31 & 37 & 43 &     2381 & 3174 & 3968 & 4762 & 5555 \\
\bottomrule
\end{tabular}
\label{tnetstatsdeg}
\begin{tablenotes}[para,flushleft]
  \footnotesize Averages over 6000 simulations. ``Size'' rows obtained from testing $H_0\colon \theta_0=\theta^*$, where $\theta^*=$ true expected value of the average degree. ``Power'' rows obtained from testing $H_0\colon \theta_0=\theta^*+0.8$.
\end{tablenotes}
\end{threeparttable}}
\end{table}

\begin{table}[ht]
\centering
\caption{Treatment Spillovers}
\begin{threeparttable}
\resizebox{\columnwidth}{!}{
\begin{tabular}{llrrrrr|rrrrr|rrrrr}
\toprule
  & $n$      & \multicolumn{5}{c}{100} & \multicolumn{5}{c}{500} & \multicolumn{5}{c}{1k} \\
\cmidrule{2-17}
  & Size &   5.55 &   5.95 &   5.83 &   5.87 &   5.88 &   4.98 &    4.33 &    5.68 &    5.20 &    5.27 &    5.05 &    4.95 &    4.93 &    4.67 &    4.65 \\
M & Power &   8.08 &  10.12 &  11.62 &  13.57 &  15.85 &  14.85 &   19.52 &   22.12 &   25.93 &   30.10 &   19.93 &   25.43 &   31.40 &   35.03 &   41.98 \\
  & $R_n$ &   6 &   8 &  10 &  12 &  14 &  13 &   18 &   22 &   26 &   31 &   19 &   25 &   31 &   37 &   43 \\
\cmidrule{2-17}
  & Size &   4.07 &   4.47 &   4.02 &   3.63 &   3.68 &   4.58 &    5.03 &    4.23 &    4.73 &    4.48 &    5.28 &    4.90 &    4.48 &    4.62 &    4.63 \\
U & Power &  16.98 &  18.88 &  20.98 &  23.45 &  25.23 &  60.25 &   68.15 &   73.77 &   78.82 &   83.28 &   88.72 &   93.77 &   96.63 &   98.35 &   99.02 \\
  & $R_n$ & 110 & 147 & 184 & 221 & 258 & 944 & 1259 & 1574 & 1889 & 2204 & 2381 & 3174 & 3968 & 4762 & 5555 \\
\bottomrule
\end{tabular}}
\label{tspill}
\begin{tablenotes}[para,flushleft]
  \footnotesize Averages over 6000 simulations. $M =$ mean-type statistic, $U =$ U-type statistic.
\end{tablenotes}
\end{threeparttable}
\end{table}

\begin{table}[ht]
\centering
\caption{Power Law Test}
\begin{threeparttable}
\begin{tabular}{lrrrrrr}
\toprule
{} & \multicolumn{3}{c}{Exponential} & \multicolumn{3}{c}{Power Law} \\
\cmidrule{2-7}
Favor Exp &       56.93 &   97.07 &    99.93 &      0.00 &    0.00 &    0.00 \\
Favor PL  &        0.00 &    0.00 &     0.00 &     60.47 &   86.30 &   93.83 \\
LL        &     -370.48 & -786.58 & -1103.83 &    301.72 &  432.20 &  498.63 \\
$R_n$     &      184 & 1574 &  3968 &    184 & 1574 & 3968 \\
$n$ &        100  &    500  &     1000 &      100  &    500  &    1000 \\
\bottomrule
\end{tabular}
\label{tPL}
\begin{tablenotes}[para,flushleft]
  \footnotesize Averages over 6000 simulations. ``LL'' $=$ average normalized log-likelihood ratio. ``Favor Exp'' $=$ \% of rejections in favor of exponential.
\end{tablenotes}
\end{threeparttable}
\end{table}


\end{document}